\documentclass[11pt]{article}
\usepackage[margin=1.2in]{geometry}
\usepackage[centertags]{amsmath}
\usepackage{amsfonts,amsthm,amssymb}
\usepackage{amssymb}
\usepackage{amsmath}
\usepackage{color}
\usepackage{relsize}

\theoremstyle{definition}
\newtheorem{theorem}{Theorem}
\newtheorem{corollary}{Corollary}
\newtheorem{definition}{Definition}

\newtheorem{proposition}{Proposition}

\newtheorem{conjecture}{Conjecture}

\begin{document}

\title{Simple approximate equilibria in games with many players}
\author{Itai Arieli\footnote{Technion, Industrial Engineering and Management  {\tt iarieli@technion.ac.il}} \ and Yakov Babichenko\footnote{Technion, Industrial Engineering and Management  {\tt yakovbab@technion.ac.il}}}

\maketitle

\begin{abstract}
We consider $\epsilon$-equilibria notions for constant value of $\epsilon$ in $n$-player $m$-actions games where $m$ is a constant. We focus on the following question: What is the largest grid size over the mixed strategies such that $\epsilon$-equilibrium is guaranteed to exist over this grid.

For Nash equilibrium, we prove that constant grid size (that depends on $\epsilon$ and $m$, but not on $n$) is sufficient to guarantee existence of weak approximate equilibrium. This result implies a polynomial (in the input) algorithm for weak approximate equilibrium.

For approximate Nash equilibrium we introduce a closely related question and prove its \emph{equivalence} to the well-known Beck-Fiala conjecture from discrepancy theory. To the best of our knowledge this is the first result introduces a connection between game theory and discrepancy theory.

For correlated equilibrium, we prove a $O(\frac{1}{\log n})$ lower-bound on the grid size, which matches the known upper bound of $\Omega(\frac{1}{\log n})$. Our result implies an $\Omega(\log n)$ lower bound on the rate of convergence of dynamics (any dynamic) to approximate correlated (and coarse correlated) equilibrium. Again, this lower bound matches the $O(\log n)$ upper bound that is achieved by regret minimizing algorithms.

\end{abstract}
\section{Introduction}
The algorithmic aspect of equilibria has been studied extensively from the moment when the concept of Nash equilibrium \cite{LH} was introduced, and mainly in the past three decades \cite{GZ,PPAD,DGP,R}. A naive approach for computation of approximate Nash equilibrium in normal form games is the following: 
\begin{itemize}
\item[-] Set a ``dense enough grid" of the strategy profiles, such that approximate Nash equilibrium is guaranteed to exist on this grid.
\item[-] Exhaustively search over all grid points whether it forms an approximate Nash equilibrium.
\end{itemize}
Despite the extensive study of equilibrium computation and the naivety of the above algorithm, no better algorithm for computation of approximate equilibrium is known (except for special classes of games, e.g., \cite{DP}). Surprisingly, the above algorithm is known to be \emph{optimal} for games with constant number of players (under the exponential hypothesis for the PPAD class), see \cite{R}. We show an additional case where this algorithm is optimal (see Corollary \ref{cor:poly}).

This motivates the study of the question: how dense should the grid of the strategy profiles be in order to guarantee existence of approximate equilibrium over the grid. A standard notion that captures the grid's density is the following.
\begin{definition}
A probability distribution $\mu \in \Delta(B)$ is called \emph{$k$-uniform} if it is uniform distribution over a multi-set of $B$ of size $k$, or equivalently, if for every $b\in B$ we have $\mu(b)=\frac{c}{k}$ for some $c\in \mathbb{Z}$. 
\end{definition}

The class of grids (over the mixed strategies) that is considered in the present paper is the class of $k$-uniform distributions, and its density is determined by $k$. The larger $k$ is, the denser the grid is.

The main question of the present paper can be formulates as follows:

\vspace{1mm}
\noindent
\textbf{Question 1:} Given an $\epsilon$-equilibrium solution concept, for which values of $k=k(\epsilon,n,m)$ existence of 
$k$-uniform $\epsilon$-equilibrium is guaranteed for every $n$-player $m$-action game? More concretely, we will be interested in understanding the asymptotic behaviour of $\lim_{n \rightarrow \infty} k(\epsilon,n,m)$ when we set $m$ and $\epsilon$ as fixed constants.

The dependence of $k=k(\epsilon,n,m)$ on the number of actions $m$ has a neat characterization of $k=\Theta(\log m)$ \cite{LMM,A}. However, the dependence on the number of players $n$ is less understood: It is known that the dependence is at most $O(\log n)$ (see \cite{BBP}), but no lower bounds were known (neither for Nash equilibria nor for correlated equilibria). The present paper aims to close these gaps. The established results have implications to equilibria computation, and to rate of convergence of learning dynamics.


\subsection{Main results}
\subsubsection{Approximate Nash equilibria}
Here, by $k$-uniform approximate Nash equilibrium we refer to an action profile where every player uses a $k$-uniform strategy. The best known upper bound for $\epsilon$-Nash equilibrium is $k=O(\log n)$, see \cite{BBP}. The question whether $k=O(1)$ suffices is an interesting open question and we address it here. In Theorem \ref{theo:awne} we prove that for the weaker notion of weak approximate equilibrium indeed $k=O(1)$ suffices. This result implies a polynomial (in the input) algorithm for computing weak approximate Nash equilibrium. To the best of our knowledge, no previous results have demonstrated the existence of polynomial algorithm for any approximate notion of Nash equilibrium in normal form games. It is interesting to note that the query complexity of weak approximate Nash equilibrium is polynomial (in the input), \cite{R}. Thus, a sub-polynomial algorithm for weak approximate equilibrium is impossible. Hence, again, as in the two-player case, the naive exhaustive search algorithm is proved to be optimal.

Unfortunately, we did not succeeded to prove or disprove whether $k=O(1)$ suffices for the standard notion of $\epsilon$-Nash equilibrium (or an $\epsilon$-well supported Nash equilibrium). However, we do gain some incites about a closely related question: whether there exists an approximate Nash equilibrium on the grid that is ``close to" an exact equilibrium.

A natural approach to prove an existence of approximate equilibrium is to search for one near-by an exact equilibrium. More concretely, consider a binary-action game ($m=2$) with an exact Nash equilibrium $x=(x_i)_{i\in [n]}\in [0,1]^n$, where $x_i$ is the probability of playing the first action. The equilibrium $x$ belongs to a $\frac{1}{k}$-cube on the grid $x\in C^k_x=[\frac{c_1}{k},\frac{c_1+1}{k}]\times ... \times [\frac{c_n}{k},\frac{c_n+1}{k}]$.  

\vspace{1mm}
\noindent
\textbf{Question 2:} For which values of $k$ it is guaranteed that for every game and every exact equilibrium $x$ one of the vertices of $C^k_x$ will be an $\epsilon$-Nash equilibrium?
\vspace{1mm}

In Proposition \ref{pro:far} we show that for $k=O(\sqrt{\log n})$ there exists a game with a binary action for each player and a unique equilibrium $x$ such that all points in $C^k_x$ are not approximate Nash equilibria. Moreover, all approximate Nash equilibria on the grid (although exist) are located ``very far" from the equilibrium: There exist players who play a certain pure strategy in the exact equilibrium and the opposite strategy in any approximate well supported Nash equilibrium on the grid.

Proposition \ref{pro:far} demonstrates that finding close approximate equilibrium for $k=O(1)$ is impossible for all games. Thus we restrict attention to particular \emph{classes} of games and ask the following question:

\vspace{1mm}
\noindent
\textbf{Question 3:} Given a class of games, is it guaranteed that for every exact equilibrium $x$ one of the vertices of $C^k_x$ will be an $\epsilon$-Nash equilibrium, for $k=O(1)$? 
\vspace{1mm}


An interesting observation is that for some classes of games answering Question 3 is (probably) mathematically very challenging: We introduce a class of games for which Question 3 is \emph{equivalent}\footnote{In fact, we prove the equivalence for some concrete instances of the Beck-Fiala conjecture of approximately-(up to a constant factor)-balanced matrices, see Section \ref{sec:disc}. By considering closely related question to Beck-Fiala conjecture where the answer is known, it is reasonable to believe that these concrete instances are "the hardest". We should note that no one have proved or disproved the Beck-Fiala conjecture for these instances.} to the well known Beck-Fiala conjecture (since 1981) in discrepancy theory \cite{BF,M,C}, see Theorem \ref{theo:dis}. To the best of our knowledge Theorem \ref{theo:dis} is the first result that establishes a connection between game theory and discrepancy theory. 

\subsubsection{Approximate correlated equilibria}
Since a correlated equilibrium is a distribution over the action profiles, here a $k$-uniform distribution means a uniform distribution over $k$ action profiles.

A $k=O(\log n)$ upper bound was proved to be sufficient for existence of $\epsilon$-correlated equilibrium, see \cite{BBP}. In fact, regret minimizing algorithms converge to an $\epsilon$-correlated equilibrium in a rate of $O(\log n)$ \cite{CBL,GR}. This provides an alternative proof for the existence of such $k$-uniform approximate correlated equilibrium. In Theorem \ref{theo:ir-lower-bound} we prove a lower bound of $k=\Omega(\log n)$. This result shows that no dynamic can converge to an approximate correlated equilibrium faster than in $\Omega(\log n)$ steps, which shows the optimality, in the rate of convergence, of the regret minimizing dynamics. We note that it was known that regret minimizing dynamics cannot converge faster than in $\Omega(\log n)$ steps. Our result shows that no dynamic at all can converge faster.

We also note that if we restrict attention to the notion of $(\epsilon,\delta)$-\emph{weak} approximate correlated equilibrium: i.e., a distribution over the profiles such that $1-\delta$ fraction of players are $\epsilon$-best replying, then there exists a $k$-uniform weak approximate correlated equilibrium for $k=O(1)$. This observation is similar to the Nash equilibrium case.

\section{Preliminaries}

An $n$-player $m$-action game consist of a set of players $[n]$, with an action set $[m]$, and a payoff function $u_i:[m]^n \rightarrow [0,1]$ for every player $i$. 
Let $\Delta(B)$ denote the set of all probability distributions over the set $B$. The set of mixed strategies of player $i$ is denoted by $\Delta([m])$. The set of correlated strategies is $\Delta([m]^n)$.
The utility function of player $i$ can be naturally extended to mixed strategy profiles $u_i:(\Delta([m]))^n \rightarrow \mathbb{R}$ and to correlated strategies $u_i:\Delta([m]^n) \rightarrow \mathbb{R}$ as the expected payoff under the given distributions. Given a profile of (mixed) actions $x=(x_i)$ we denote by $x_{-i}$ the profile of actions of player's $i$ opponents, namely $x_{-i}=(x_1,...,x_{i-1},x_{i+1},...,x_n)$.

\begin{definition}
\textbf{Approximate Nash equilibrium.} A profile of mixed actions $(x_i)_{i\in [n]}$ where $x_i\in \Delta([m])$ is an \emph{$\epsilon$-Nash equilibrium} if no player can gain more than $\epsilon$ by a unilateral deviation. Namely, $u_i(x_i,x_{-i})\geq u_i(a_i,x_{-i})-\epsilon$ for every player $i$ and every action $a_i\in [m]$.
\end{definition}

\begin{definition}
\textbf{Weak approximate Nash equilibrium.} A profile of mixed actions $(x_i)_{i\in [n]}$ where $x_i\in \Delta([m])$ is an \emph{$(\epsilon,\delta)$-weak approximate Nash equilibrium} if at least $1-\delta$ fraction of the players cannot gain more than $\epsilon$ by a unilateral deviation.
\end{definition}

\begin{definition}
\textbf{(Approximately) Individually rational payoffs.} The \emph{individually rational level} of player $i$ is the maximal number $v_i$ that he can guarantee (using mixed strategies) against any action of the opponents. Namely $$v_i = \max_{x_i \in \Delta([m])} \min_{a_{-i}\in [m]^{n-1}} u_i(x_i,a_{-i}).$$
A correlated distribution $x\in \Delta([m]^n)$ is \emph{$\epsilon$-individually rational} if $u_i(x)\geq v_i - \epsilon$.
\end{definition}

There are several notions of correlated equilibria and approximate correlated equilibria, arguably the strongest notion of approximate correlated equilibrium is the following.

\begin{definition}
\textbf{Approximate correlated equilibrium.} 
A correlated distribution $x\in \Delta([m]^n)$ is an \emph{$\epsilon$-correlated equilibrium} if $$\mathbb{E}_{a\sim x} [u_i(a)]\geq \mathbb{E}_{a\sim x} [u_i(f(a_i),a_{-i})]-\epsilon$$ for every player $i$ and every switching function $f:[m] \rightarrow [m]$.
\end{definition}
The intuition behind this notion comes from the idea that a correlated strategy $x$
can be implemented by a mediator who draws an action profile $a=(a_i)_{i\in [n]}$ according to $x$ and recommends every player $i$ to use $a_i$. If no player can gain more than $\epsilon$ by deviating from mediator's recommendation (namely to play action $f(j)$ every time mediator recommends $j$) then the distribution $x$ is an $\epsilon$-correlated equilibrium.

Approximate correlated equilibrium also has a weak analogue of the solution.
\begin{definition}
\textbf{Weak approximately correlated equilibrium.} 
A correlated distribution $x\in \Delta([m]^n)$ is an \emph{$(\epsilon,\delta)$-weak approximate correlated equilibrium} if $$\mathbb{E}_{a\sim x} [u_i(a)]\geq \mathbb{E}_{a\sim x} [u_i(f(a_i),a_{-i})]-\epsilon$$ for every switching function $f:[m] \rightarrow [m]$ for at least $(1-\delta)$-fraction of the players.
\end{definition}

We would like to note that $\epsilon$-individual rationality is the weakest possible notion for solutions that require rationality from all players (unlike weak approximate correlated equilibrium for instance). In particular the set of approximately individually rational distributions contains the set of approximate correlated equilibria and approximate coarse correlated equilibria.
 
\section{Weak approximate Nash equilibrium}
\begin{theorem}\label{theo:awne}
Every $n$-player $m$-actions game admits a $k$-uniform $(\epsilon,\delta)$-weak approximate Nash equilibrium for every $k\geq\frac{32(\ln 8 + \ln m -\ln \epsilon -\ln \delta)}{\epsilon^2}$.
\end{theorem}
The important property of the bound on $k$ is the fact that it \emph{does not} depend on $n$. A straightforward corollary from Theorem \ref{theo:awne} is the existence of polynomial algorithm for weak approximate equilibrium.

\begin{corollary}\label{cor:poly}
For constant $\epsilon,\delta>0$ and $m\in \mathbb{N}$ there exists a $poly(N)$ algorithm that computes an $(\epsilon,\delta)$-weak approximate Nash equilibrium in every $n$-player $m$-action game, where $N=n\cdot m^n$ is the input size (the size of the game). 
\end{corollary} 
\begin{proof}[Proof of Corollary \ref{cor:poly}]
The algorithm exhaustively search for an $(\epsilon,\delta)$-weak approximate Nash equilibrium over all the $k$-uniform profiles, for $k=\frac{32(\ln 8 + \ln m -\ln \epsilon -\ln \delta)}{\epsilon^2}$. The number of $k$-uniform mixed actions of a single player is bounded by $k^m$. Thus, the number of $k$-uniform action profiles is bounded by 
\begin{align*}
(k^m)^n=(m^n)^{\frac{m\log k}{\log m}}\leq N^{\frac{m\log k}{\log m}}=poly(N)
\end{align*}
Note that the algorithm is guaranteed to find an $(\epsilon,\delta)$-weak approximate Nash equilibrium by Theorem \ref{theo:awne}.  
\end{proof}

The proof of Theorem \ref{theo:awne} uses similar technique to the one developed in \cite{BBP}. We prove that after \emph{constant} number of samples from an exact Nash equilibrium distribution, with positive probability the sampled mixed action profile forms an $(\epsilon,\delta)$-weak approximate Nash equilibrium. We rely on the following concentration inequality for product distributions that was derived in \cite{BBP}. 

Given a discrete probability space $(\Omega,\mu)$ we denote by $\mu^{(k)}\in \Delta(\Omega)$ the random distribution that is obtained by taking the average of $k$ i.i.d. samples from $\mu$. Namely, $\mu^{(k)}(\omega)=\frac{1}{k}\sum_i \mathbf{1}_{x_i=\omega}$ when $x_1,...,x_k\sim \mu$ are i.i.d. random variables.

\begin{theorem}[\cite{BBP}]\label{theo:prod-con}
Let $(\Omega_1,\mu_1),...,(\Omega_n,\mu_n)$ be discrete probability spaces. Consider the product space $(\Omega=\Pi_i \Omega_i , \mu = \Pi_i \mu_i)$. For every $\hat{\epsilon}>0$, $k\in \mathbb{N}$, and $f:\Omega \rightarrow [0,1]$ we have
\begin{equation*}
\mathbb{P}(|\mathbb{E}_{\Pi_i \mu_i^{(k)}}[f]-\mathbb{E}_\mu [f]|>\hat{\epsilon}) \leq \frac{4e^{-(\hat{\epsilon}^2 /8)k }}{\hat{\epsilon}}
\end{equation*}
\end{theorem}

\begin{proof}[Proof of Theorem \ref{theo:awne}]
Let $x=(x_i)_{i\in [n]}$ be a Nash equilibrium of the game. We denote $s_i^k=x_i^{(k)}$ the mixed action of player $i$ that is obtained by sampling $k$ i.i.d. draws from $x_i$.
 
Setting $f=u_i$ and $\hat{\epsilon}=\frac{\epsilon}{2}$ in Theorem \ref{theo:prod-con} implies that
\begin{align*}
\mathbb{P}(|u_i(a_i,s^k_{-i})-u_i(a_i,x_{-i})|\geq \epsilon)\leq  \frac{8e^{-(\epsilon^2 /32)k }}{\epsilon}
\end{align*}
for every player $i\in [n]$ and every action $a_i\in [m]$. The choice of $k$ guarantees that
\begin{align*}
\frac{8e^{-(\epsilon^2 /32)k }}{\epsilon}<\frac{\delta}{m}.
\end{align*}
Using the union bound, we get that for every player $i$ with probability greater than $1-\delta$ we have $|u_i(a_i,s_{-i}^k)-u_i(a_i,x_{-i})|\leq \frac{\epsilon}{2}$ for \emph{all} actions $a_i\in [m]$. We denote the above event by $g_i$. Note that the event $g_i$ is a sufficient condition for player $i$ to $\epsilon$-best reply at the profile $(s_i^k)_i$, because
\begin{align*}
u_i(a_i,s_{-i}^k)\leq u_i(a_i,x_{-i})+\frac{\epsilon}{2} \leq \sum_{a'_i\in A_i} s_i^k(a'_i)u_i(a'_i,x_{-i})+\frac{\epsilon}{2}\\
\leq \sum_{a'_i\in A_i} s_i^k(a'_i)u_i(a'_i,s^k_{-i})+\epsilon=u_i(s^k_i,s^k_{-i})+\epsilon.
\end{align*}

Since each one of the events $g_i$ happen with probability of at least $1-\delta$ there exists a realization $(s_i^k)_i$ such that at least $(1-\delta)$ fraction of the events $g_i$ happen. Such a realization $(s_i^k)_i$ is an $(\epsilon,\delta)$-weak approximate Nash equilibrium.  
\end{proof}

\section{Approximate Nash equilibrium}
We consider games with binary actions where player's mixed strategy is a real number $x_i\in [0,1]$.
For clarity of presentation, we assume that $k$ is odd, and therefore the mixed strategy $\frac{1}{2}$ is $\frac{1}{2k}$-away from both closest points on the grid: $\frac{k-1}{2k}$ and $\frac{k+1}{2k}$. This assumption is for clarity of presentation only: For even values of $k$ we can slightly change the payoffs in the constructed games such that the exact equilibrium will appear at the point $\frac{k+1}{2k}$. Here, again, we have a situation where the two closest points on the grid, $\frac{1}{2}$ and $\frac{k+2}{2k}$, are $\frac{1}{2k}$ away from the exact equilibrium.

We recall that given an equilibrium $x$ we denote by  $C^k_x=[\frac{c_1}{k},\frac{c_1+1}{k}]\times ... \times [\frac{c_n}{k},\frac{c_n+1}{k}]$ the $\frac{1}{k}$-cube where the equilibrium $x$ is located.  

\begin{proposition}\label{pro:far}
There exists an $n$-player binary-action with a unique exact Nash equilibrium $x$, such that for $k<\frac{\sqrt{\log n}}{8}$ non of the points on the grid that are close to $x$ (namely the set $C^k_x$) form a $0.1$-Nash equilibrium.  
\end{proposition}
 
\begin{proof}
Consider the following game with $2b+{2b\choose b}$ players. $b$ pairs of players are playing matching-pennies with each other, we denote these players by $1,2,...,2b$ and call them the \emph{matching-pennies players}. Each one of the remaining ${2b\choose b}$ is characterized by a set $S \subset [2b]$, and we call them the \emph{observing players}. Player $S$ has to guess whether the amount of players that will play $1$ in $S$ is close to $\frac{b}{2}$. More formally, player $S$ has two strategies $0$ and $1$. His utility is given by
\begin{align*}
u(0,a_S)&=0.5 \text{ (independently of $a_S$), }\\
u(1,a_S)&=\begin{cases} 
1 \text{ if } \frac{b}{2}-2\sqrt{b} \leq \sum_{i\in S} a_i \leq \frac{b}{2}+2\sqrt{b} \\
0 \text{ otherwise. }   
\end{cases}
\end{align*}
The unique exact equilibrium of this game is where all matching-pennies players are playing $(\frac{1}{2},\frac{1}{2})$ and all the observing players are playing $1$. The latter follows from the fact that the amount of $1$'s is distributed according to a binomial distribution $Bin(b,\frac{1}{2})\approx \mathcal{N}(\frac{b}{2},\frac{b}{4})$. The probability that this amount will be in the segment $[\frac{b}{2}-\sqrt{b} , \frac{b}{2}+\sqrt{b}]$ is $1-2\Phi(-2)>0.94$.

Consider now the case where the matching-pennies players are restricted to actions on an odd grid of size $k=\frac{\sqrt{b}}{4}$. More importantly, every player is playing a mixed strategy that is $\frac{2}{\sqrt{b}}$-away from $\frac{1}{2}$. There exists a set $S$ such that all players in $S$ are either (all) playing a mixed strategy above $\frac{1}{2}$ or (all) are playing a mixed strategy below $\frac{1}{2}$. W.l.o.g., assume that the latter happen. Note that the amount of $1$s for this specific set $S$ is a distribution with expectation of at least $\mu \geq b(\frac{1}{2}+\frac{4}{\sqrt{b}})=\frac{b}{2}+2\sqrt{b}$, and a variance of at most $\sigma^2 \leq \frac{b}{4}$. By the Central Limit Theorem\footnote{Here and in the proof of Theorem \ref{theo:dis} we rely on the powerful Central limit Theorem. Similar arguments can be done by using the less powerful Chebishev's inequality.}, the probability that the amount of $1$'s are in the segment $(-\infty, \frac{b}{2}+\sqrt{b}]$ is at most $\Phi(-2)<0.03$. Therefore player $S$ must play $1$ with probability of at most $\frac{0.1}{0.47}<0.22$ in a $0.1$-Nash equilibrium, which is far from his pure strategy 1 in the exact equilibrium. 

To conclude, the total number of players in the game is $n=2b+{2b\choose b}\leq 2^{2b}$, and for odd grid of size $k<\frac{\sqrt{\log(n)}}{8}$ all approximate equilibria on the grid are far from the unique exact equilibrium.
\end{proof}

\subsection{The connection to discrepancy theory}\label{sec:disc}
The basic question that is considered in discrepancy theory is the following: Given a $0,1$ matrix $M$ of size $n\times m$, how close to $0$ can one make the sum of (all) rows by multiplying the columns by\footnote{The same question has an elegant equivalent formulation through two-coloring of elements in set systems.} $\{1,-1\}$. More formally, we define

\begin{align*}
disc(M)=\min_{\chi \in \{1,-1\}^m} ||M\chi ||_\infty.
\end{align*}

The classical ''Six standard deviations suffice" Theorem by Spencer \cite{S} states that for $m\geq n$ and for every matrix $M$ we have $disc(M)\leq 6\sqrt{n}$. For the case where the matrix is \emph{sparse} namely each column contains at most $t$ $1$s Back and Fiala conjectured that $disc(M)=O(\sqrt{t})$  (independently of $n$ and $m$) \cite{BF}.

A particular case of Back-Fiala conjecture is the case of balanced matrices (up to a constant factor).

Given a matrix $M\in \{0,1\}^{n\times m}$ we denote $r_i=\sum_{j\in [m]} M_{i,j}$ for every $i\in [n]$ and $c_j=\sum_{i\in [n]} M_{i,j}$ for every $j\in [m]$ the number of 1s in each row and column.

\begin{definition}
A matrix $M\in \{0,1\}^{n\times m}$ is called \emph{$\alpha$-balanced} if $\frac{1}{\alpha}\leq \frac{c_i}{c_j} \leq \alpha$, $\frac{1}{\alpha}\leq \frac{r_i}{r_j} \leq \alpha$, and $\frac{1}{\alpha}\leq \frac{r_i}{c_j} \leq \alpha$ for every pair of rows/ columns. 
\end{definition}

\begin{conjecture}\label{con:BF}
\textbf{$\alpha$-balanced Beck-Fiala conjecture.} For every $\alpha$-balanced matrix $M$ holds $disc(M)=O(\sqrt{t})$ where $t=r_1$ (or alternatively the sum of any other row or column).  
\end{conjecture}
There are evidences to believe that the $\alpha$-balanced conjecture is not significantly simpler than the original conjecture, because for similar problems where tight lower bounds are known (for instance Spencer's Theorem \cite{S}) the lower bounds satisfy balanceness.

We consider the class of \emph{majority matching-pennies games}. A majority matching-pennies game is characterized by a matrix $M\in \{0,1\}^{n\times m}$, and it consists of $n+m$ players. Each row/column player has two actions $\{1,-1\}$, and he should decide whether to multiply all the elements in his row/column by $-1$.  The utility of every row player is given by:
\begin{itemize}
\item $1$ if the sum of numbers in his row is positive (after we take into account his action and the actions of his column opponents).
\item $0$ if the sum of numbers in his row is $0$.
\item $-1$ if the sum of numbers in his row is negative.
\end{itemize}

The utilities of every column player is the opposite of the row players:

\begin{itemize}
\item $-1$ if the sum of numbers in his row is positive.
\item $0$ if the sum of numbers in his row is $0$.
\item $1$ if the sum of numbers in his row is negative.
\end{itemize}

We call these games majority matching-pennies since they are equivalent to a bipartite polymatrix game where row player $i$ interacting with column player $j$ in a matching pennies game iff $M_{i,j}=1$. Unlike standard polymiatrix games where players receive the sum of payoffs, here players receive the sign of the sum of payoffs.

Note that the profile of actions where every player is playing $(\frac{1}{2},\frac{1}{2})$ is an exact equilibrium in every majority matching-pennies game. 

A subclass of majority matching-pennies games is \emph{$\alpha$-balanced majority matching pennies games} where the matrix $M$ is $\alpha$-balanced.

\begin{conjecture}\label{con:close-eq}
\textbf{Existence of close approximate equilibrium equilibrium, for $\alpha$-balanced games.} There exists a global constant $k=k(\alpha)$ such that one of the $2^{n+m}$ profiles of the form $(\frac{k \pm 1}{2k})^{n+m}$ is a $0.4$-Nash equilibrium. 
\end{conjecture}

Now we are ready to state the equivalence of approximate equilibria and the discrepancy result.

\begin{theorem}\label{theo:dis}
For every constant $\alpha$, Conjecture \ref{con:BF} is equivalent to Conjecture \ref{con:close-eq}.
\end{theorem}

Although Theorem \ref{theo:dis} demonstrates a connection between Back-Fiala conjecture and a very specific question on approximation of equilibrium, we believe that the Theorem provides interesting incites for the following reasons. First, as mentioned in the introduction, it connects between two seemingly unrelated topics of game theory and discrepancy theory. Second, it suggests that the question of the existence of $k$-uniform $\epsilon$-Nash equilibrium for $k=O(1)$ is probably quite involved. Third, if one comes up with a proof of sufficiency of $k=O(1)$ for the existence of $k$-uniform $\epsilon$-Nash equilibrium in binary-action games, then it would be interesting to understand where this approximate Nash equilibrium in majority matching-pennies games is located. Informally, note that majority matching-pennies games have conflict of interests between the row players and the column players: row players want the matrix to increase the number of $1$s whereas column players want to increase the number of $(-1)$s. Thus it is reasonable to believe that the approximate Nash equilibrium will indeed be located ``close to" (at least in the relative $l_1$ metric) the exact equilibrium $(\frac{1}{2})^{n+m}$.

\begin{proof}[Proof of Theorem \ref{theo:dis}]
First we prove that Conjecture \ref{con:BF} implies Conjecture \ref{con:close-eq}. 

Beck-Fiala conjecture implies that for every matrix $M$ there exists $\chi \in \{-1,1\}^m$ such that $||M\chi ||_\infty \leq C\sqrt{t}$. Note that $M^T$ is also an $\alpha$-balanced matrix, thus there exists also $\chi' \in \{-1,1\}^n$ such that $||M^T \chi' ||_\infty \leq C\sqrt{t}$.

We set $k=3\alpha C$, and we argue that the profile of actions where every column player $j$ is playing $\frac{k+\chi_j}{2k}$, and every row player $i$ is playing $\frac{k+\chi'_j}{2k}$ is a $0.4$-Nash equilibrium.

We prove that the first row player cannot gain more than $\frac{1}{3}$ by deviation. For other players the same arguments hold. Given that the first row player is playing 1, we analyse the distribution $S$ of the sum of elements at the first row (given the mixed strategy of the column players). $S$ is a sum of $t'$ Bernoulli $\pm 1$ variables for $\frac{1}{\alpha}t < t' <\alpha t$. The expectation of $S$ is exactly $\mu=\mathbb{E}[S] = \sum_{j} \frac{\chi_j}{k}M_{1,j}$. In addition $\sigma^2\geq \frac{1}{4} \frac{t}{\alpha}$. By the central limit Theorem the distribution $S$ can be approximated by $S\sim N(\mu,\sigma^2)$. Note that by Beck-Fiala conjecture we have $| \mu | \leq \frac{C\sqrt{t}}{k}$.  Wlog, assume that $\mu\geq 0$. Then we have that $Pr(S<0)\approx \Phi(-\frac{\mu}{\sigma})\geq \Phi(-\frac{2C\alpha}{k})\geq 0.4$. So by playing $1$ the first row player receives a payoff of at most $0.2$ and at least $0$. Therefore by playing $-1$ the first row player  receives a payoff of at least $-0.2$ and at most $0$. So every mixed strategy is at least $0.4$-best reply.

Now we prove the opposite direction, that Conjecture \ref{con:close-eq} implies Conjecture \ref{con:BF}.

Given a matrix $M$, we consider the mixed action profile $(\frac{k \pm 1}{2k})^{n+m}$ that forms a $0.4$-Nash equilibrium in the majority matching-pennies game that is defined by $M$, and we set $\chi_i=\pm 1$ for $i\in [m]$ according to the mixed strategy of the $i$th column player at the $0.4$-Nash equilibrium. Namely, if the $i$th column player is playing $\frac{k+1}{2k}$ we set $\chi_i = 1$, otherwise, if he plays $\frac{k-1}{2k}$, we set $\chi_i = -1$. We show that the first row sums up to $O(\sqrt{t})$. For every other row the same arguments hold. Note that the first row player plays a mixed strategy in $[0.4,0.6]$. Therefore, in a $0.4$-Nash equilibrium he must be $1$-indifferent between the two actions $1$ and $-1$. As before, denote by $S$ the sum of elements at the first row in the $0.4$-Nash equilibrium profile, given that the first row player  is playing 1. 1-indifference implies that $0.25\geq Pr(S\geq 0)$ and $0.25 \geq Pr(S \geq 0)$. The variance of $S$ is at most $\sigma^2 \leq \alpha t$. If $\mu > \sqrt{\alpha t}$, then by the Central Limit Theorem we get that $Pr(S\leq 0)\leq 0.16$ which is a contradiction. Similarly it is impossible that $\mu < -\sqrt{\alpha t}$. Therefore, $|\mu|<\sqrt{\alpha t}$ which implies that $|\sum_j \frac{\chi_j}{k} M_{1,j} | \leq \sqrt{\alpha t}$. Namely $|\sum_j \chi_j M_{1,j} | \leq k \sqrt{\alpha t}=O(\sqrt{t})$.
\end{proof}


\section{Correlated distributions}

\begin{theorem}\label{theo:ir-lower-bound}
There exists an $n$-player binary-actions game where every distribution with support $k<\frac{1}{4}\log n$ is not $\frac{1}{4}$-individually rational.  
\end{theorem}

Since every $\epsilon$- correlated/ coarse correlated equilibrium distribution is, in particular, $\epsilon$-individually rational, the above lower bound holds for correlated and coarse correlated equilibria. 

We recall that a dynamic is \emph{converging to an approximate correlated/ coarse correlated equilibrium} in $t$ steps if the empirical distribution of play forms an approximate  correlated/ coarse correlated equilibrium (with high probability). 

\begin{corollary}
No dynamic converges to approximate  correlated/ coarse correlated equilibrium faster than in $\Omega(\log n)$.
\end{corollary}
The corollary simply follows from the fact that no empirical distribution of size smaller than $\Omega(\log n)$ can form an approximate correlated/ coarse correlated equilibrium. On the other hand, it is known that regret minimizing dynamics converge to approximate correlated/ coarse correlated equilibrium in $O(\log n)$ steps, see \cite{CBL,GR}. Thus, $\Theta(\log n)$ is the fastest possible rate of convergence of dynamics to approximate correlated 
correlated/ coarse correlated equilibrium and regret minimizing dynamics achieve this bound.

\begin{proof}[Proof of Theorem \ref{theo:ir-lower-bound}]
For two vectors $x,y\in \{0,1\}^k$ we denote $x\oplus y = x+y \mod 2\in  \{0,1\}^k$.

For every $k\in \mathbb{N}$ we consider a game with $n=(2^k-1) 2^{(2^{k-1})}$ players. Players are denoted by $(s,p)$ where $0\neq s\in \{0,1\}^k$ and $p\in \{0,1\}^{(2^{k-1})}$. We identify a player with a pair $(s,V)$ where $V\subset \{0,1\}^k$ is a subset of size $|V|=\frac{1}{2} 2^{k}$ as follows: note that $s$ defines a matching over $\{0,1\}^k$ $x\leftrightarrow (x\oplus s)$ with $2^{k-1}$ matched pairs. The vector $p$ indicates which one of the two vectors ($x$ or $(x\oplus s)$) belongs to the subset $V$.

Each player has binary actions set $\{0,1\}$. We define a mapping $f:\{0,1\}^n\rightarrow \{0,1\}^k$ which assigns a $k$-dimensional binary vector for each action profile in the game:
\begin{align*}
f(a)=\underset{(s,p):a_{(s,p)}=1 }{\mathlarger{\mathlarger{\mathlarger{\oplus}}}} s. 
\end{align*}
Now we define the utility of player $(s,p)$, or equivalently player $(s,V)$ to be 
\begin{align*}
u_{(s,p)}(a)=\begin{cases}
1 \text{ if } f(a)\in V \\
0 \text{ otherwise.}
\end{cases}
\end{align*}

Note that any unilateral divination of a single player at any action profile switches his utility from 0 to 1, or from 1 to 0. Formally, for every player 
$(s',p')$ (or equivalently $(s',V')$) and every action profile of the opponents $a_{-(s',p')}$ holds $$(u_{(s',p')}(0,a_{-(s',p')}),u_{(s',p')}(1,a_{-(s',p')}))\in \{(0,1),(1,0)\},$$ because
\begin{align*}
u_{(s',p')}(0,a_{-(s',p')})=0 \Rightarrow \underset{(s,p):a_{(s,p)}=1 }{\mathlarger{\mathlarger{\mathlarger{\oplus}}}} s \notin V' \Rightarrow \underset{(s,p):a_{(s,p)}=1 }{\mathlarger{\mathlarger{\mathlarger{\oplus}}}} s \oplus s' \in V' \Rightarrow u_{(s',p')}(1,a_{-(s',p')})=1.
\end{align*}
Similarly, $u_{(s',p')}(0,a_{-(s',p')})=1$ implies that $u_{(s',p')}(1,a_{-(s',p')})=0$. 

Therefore, each player can guarantee a payoff of $\frac{1}{2}$ by playing $(\frac{1}{2},\frac{1}{2})$, which shows that the individually rational level of every player is at least $\frac{1}{2}$.

We should show that for every distribution $\mu \in \Delta(\{0,1\}^n)$ with support of at most $\frac{1}{4}2^{k}\geq \frac{1}{4}\log n$, there exists a player who receives a payoff of at most $\frac{1}{4}$ (which is $\frac{1}{4}$-far from his individually rational level). 

The distribution $\mu \in \Delta(\{0,1\}^n)$ induces a distribution $\nu =f(\mu)\in \Delta(\{0,1\}^k)$. The support of $\nu$ remains to be at most $\frac{1}{4}2^{k}$. We apply the probabilistic method to find a player with low payoff. We choose a vector $x$ at random according to $\nu$, and we choose a vector $s$ (independently) uniformly at random from $\{0,1\}^k$. Note that for every fixed vector $x$, the vector $x\oplus s$ is distributed uniformly at random, therefore $Pr_{x,s}(x \oplus s \notin support(\nu))\geq \frac{3}{4}$. Therefore, there exists a vector $s'$ such that $Pr_{x \sim \nu}(x \oplus s' \notin support(\nu))\geq \frac{3}{4}$. Consider the matching of $\{0,1\}^k$ that is obtained by the vector $s'$. We set the vector $p'$ as follows: For a pair $x\leftrightarrow x \oplus s'$ where $x\in support(\nu)$ and $x \oplus s' \notin support(\nu)$ we set $x \oplus s' \in V'$ (and thus $x\notin V'$). For all other pairs we set the choice of $p'$ arbitrarily. By definition, the payoff of player $(s',p')$ can be expressed as $u_i(\mu)=Pr_{x\sim \nu} (x \in V')$, but we have set $V'$ in a way that guarantees $Pr_{x\sim \nu} (x \in V')\leq \frac{1}{4}$.

\end{proof}

In contrast to approximate correlated equilibrium which requires support of size $k=\Omega(\log n)$, for the weaker notion of weak approximate correlated equilibrium where we allow a small constant fraction of players to have an arbitrary regret, existence of $k$-uniform weak approximate equilibrium is guaranteed for $k=O(1)$.

\begin{proposition}
Every $n$-player $m$-actions game admits a $k$-uniform $(\epsilon,\delta)$-weak approximate correlated equilibrium for\footnote{In fact, a polylogarithmic dependence on $m$ can also be obtained, using slightly more involved arguments, as in Theorem 6 in \cite{BBP}.}  $k=2\frac{m\ln m -\ln \delta}{\epsilon^2}$.
\end{proposition}

The proof is similar to the proof of Theorem 5 in \cite{BBP}. In \cite{BBP} Theorem 5 they show that after $k=O(\log n)$ samples from an exact correlated equilibrium \emph{all} players will have low regret (w.h.p). Here we observe that after $k=O(1)$ samples \emph{most} of the players will have low regret (w.h.p).

\begin{proof}[Sketch of the proof]
We sample $k$ samples from an exact correlated equilibrium, and we consider the regret of a single player $i$ for not using the switching rule $f:[m] \rightarrow [m]$ (namely, every time player $i$ was recommended to play action $j$ he switches to $f(j)$). The probability that the this regret will exceed $\epsilon$ is $e^{-\Theta(\frac{k}{\epsilon^2})}$. Denote by $g_i$ the event where for player $i$ the regret is be below $\epsilon$ for all switching rules. By the union bound $Pr(g_i)\geq 1-m^m e^{-\Theta(\frac{k}{\epsilon^2})}$. Therefore, there exists a realization for which at least $1-m^m e^{-\Theta(\frac{k}{\epsilon^2})}$ fraction of the events $g_i$ occur. By the choice of $k$ we have $1-m^m e^{-\Theta(\frac{k}{\epsilon^2})}>1-\delta$. Namely, this realization of the sampling forms an $(\epsilon,\delta)$-weak approximate equilibrium.
\end{proof}

\end{document}